\documentclass[11pt]{article}

\usepackage[utf8]{inputenc}
\usepackage[T1]{fontenc}
\usepackage[a4paper, margin=3cm]{geometry}
\usepackage{amsmath, amssymb, amsthm}
\usepackage{bbm}
\usepackage[disable]{todonotes}
\usepackage{hyperref}
\usepackage[thmtools-compat]{keytheorems}

\theoremstyle{plain}
\newtheorem{theorem}{Theorem}[section]

\newtheorem{proposition}[theorem]{Proposition}

\newtheorem{conjecture}[theorem]{Conjecture}
\newtheorem{observation}[theorem]{Observation}

\theoremstyle{definition}
\newtheorem{definition}[theorem]{Definition}
\newtheorem{example}[theorem]{Example}

\theoremstyle{remark}

\newcommand{\BibTeX}{\rm B\kern-.05em{\sc i\kern-.025em b}\kern-.08em\TeX}

\DeclareMathOperator*{\argmax}{arg\,max}

\newcommand{\agents}{\mathcal{N}}
\newcommand{\items}{\mathcal{G}}
\newcommand{\allocations}{\mathcal{A}}
\newcommand{\bid}{\mathbbm{r}}
\newcommand{\indbid}{r}

\newcommand{\bids}{\mathcal{L}(\items)}

\title{On the Convergence of Iterated Round Robin}

\author{
  Maxime Lucet \\ \small LIP6, CNRS, Sorbonne Université, Paris, France \\ \small\texttt{maxime.lucet@lip6.fr}
  \and
  Henri Surugue \\ \small LIP6, CNRS, Sorbonne Université, Paris, France \\ \small\texttt{henri.surugue@lip6.fr}
  \and
  Nawal Benabbou \\ \small LIP6, CNRS, Sorbonne Université, Paris, France \\ \small\texttt{nawal.benabbou@lip6.fr}
  \and
  Aur\'elie Beynier \\ \small LIP6, CNRS, Sorbonne Université, Paris, France \\ \small\texttt{aurelie.beynier@lip6.fr}
  \and
  Nicolas Maudet \\ \small LIP6, CNRS, Sorbonne Université, Paris, France \\ \small\texttt{nicolas.maudet@lip6.fr}
  \and
  Arianna Novaro \\ \small Université Paris 1 Panthéon-Sorbonne, CNRS, CES, Paris, France \\ \small\texttt{arianna.novaro@univ-paris1.fr}
}

\date{}

\begin{document}

\maketitle

\begin{abstract}
Round Robin (RR) is a simple and widely used picking sequence mechanism for allocating indivisible goods, known to guarantee envy-freeness up to one item (EF1) when preferences are additive. Despite being manipulable (meaning that agents may have an incentive to strategically misreport their preferences), it is known to always admit at least one pure Nash equilibrium. In this paper, we study an iterated variant of RR, in which agents observe the outcome of the mechanism at each iteration and may revise their reported preferences through a best-response dynamics.

We first consider a full-information setting, where agents observe both the outcome and the reported preferences of all the agents. In this setting, we show that cycles can occur even under strong restrictions on the manipulations carried out by the agents. We nonetheless identify a natural restriction (\emph{item-preserving swaps}), under which convergence to a pure Nash equilibrium is guaranteed to exist for any number of agents and items. We also show that restricting to \emph{adjacent swaps} guarantees convergence for two agents and a limited number of items, leaving the question open for unbounded number of items.
We then turn to a partial-information setting, in which agents only observe the resulting allocation and must then infer the preference profiles of the others assuming consistency with the RR sequence. We study three behaviours guiding the agents when deciding how to manipulate under uncertainty on the other agents' preferences. We show that two of the studied behaviours, despite being quite natural, may cycle while a more cautious rule guarantees iterated RR to be strategyproof.
\end{abstract}

\section{Introduction}

Picking sequences \cite{BL-IJCAI11} are simple mechanisms designed to allocate indivisible goods to agents, a central topic in multiagent systems \cite{AMANATIDIS2023103965}: following a pre-defined sequence, agents are selected in turn, possibly multiple times, to pick their preferred object among the remaining ones, until no item remains. It has low communication requirement, as it only requires agents to rank the items \cite{LangR16}. We make the common assumption that preferences are provided upfront, and that the mechanism then computes the allocation without further interacting with the agents.  
Unfortunately, it is well-known that, apart from a very restrictive subclass of sequences (essentially extensions of serial dictatorships \cite{MBGL-IJCAI25}), manipulation can occur. For instance, if an item that I value highly is poorly valued by all other agents, I may not need to reveal my true preference to obtain it, since no one will compete for it. 
I may instead misreport my preferences so as to secure another desirable item that would otherwise be taken by another agent. The complexity of finding such manipulations has been studied, see {\it e.g.}  \cite{Kalinowski2013,Aziz2017}. 

Round Robin (RR) is a specific picking sequence in which agents are repeatedly called in the same order until exhaustion of objects. Despite its  simplicity, it has remarkable properties when agents' preferences are additive: it guaranties envy-freeness up-to-one-item (EF1) \cite{Caragiannis_et_al2019}, it provides some form of efficiency (even if weaker than Pareto-optimality \cite{beynier:hal-02132742}), and it is known to always have pure Nash equilibria (PNE) when agents behave strategically (that is, misreport their preferences over items, as discussed above) \cite{Aziz_et_al2017}. Furthermore, it preserves this EF1  guarantee in PNE \emph{with respect to the true values} of the items \cite{AmanatidisBFLLR24}. 
This is a striking result and a strong motivation to investigate further the behaviour of RR when used by strategic agents. In particular, it is not clearly understood under what circumstances convergence to a PNE may occur.

In this paper, we study an iterated version of RR, called Iterated Round Robin (IRR), and identify conditions under which it converges to a PNE. Such a mechanism may be useful when allocating resources within an organization, where agents may wish to revise their announced preferences after observing the resulting allocation. For instance, when allocating time slots for access to shared facilities among research teams, teams could be allowed to successively revise their preferences in response to the current proposed allocation, until no one wishes to change. More generally, such iterative procedures can be viewed as a way of reaching a form of stability, in the sense that the process stops when no agent wishes to revise her report. Note that items are only actually allocated at the very end of this process (this differs from another line of research studying \emph{repeated} allocations, see {\it e.g} \cite{Igarashi2024}). 
A related model, known as \emph{iterated voting}, has been studied: in this setting, voters repeatedly revise their ballots and a central question is whether this process converges to a stable outcome \cite{MeirPRJ10,LevR16,MeirPRJ17,TerzopoulouEtAlAIJ2024}.

When defining such iterated mechanisms, it is important to specify (i) the \emph{dynamics} of the process (which agents can react, and when), (ii) the \emph{manipulations} agents are able to perform, and (iii) the \emph{level of information} available to them when making their decisions. For (i), we adopt standard best-response dynamics, where agents update their bids individually and sequentially.
Regarding (ii), we consider both the unrestricted case and several cognitively plausible restrictions on how agents can modify their previous bid at each turn. Finally, for (iii), we assume throughout the paper that the mechanism itself, i.e., the exact picking sequence, is known to all agents, and consider two levels of information. Under full information, agents observe the full bid profile and hence the resulting allocation. Under partial information, they observe only the resulting allocation and can therefore only infer an information set of possible bid profiles consistent with it. It has already been observed that under unrestricted manipulations in the full information setting, cycles may occur under a best response dynamics \cite{Aziz_et_al2017}.

Our paper provides a more detailed and nuanced picture. In the full-information setting (Section~\ref{section: full info}), we show that cycles may occur even under severe restrictions on manipulations, but also identify non-trivial cases guaranteeing convergence to PNE. When manipulations are restricted, the set of potential PNE is enlarged, so the guarantee that ``PNE are EF1 with respect to the true valuations'' no longer necessarily holds; indeed, we show that it may fail.
In the partial-information setting (Section~\ref{section: partial info}), we distinguish different behaviours under uncertainty about the preferences of the other agents. We show that some provide strategyproofness guarantees, while others may lead to manipulations and cycles.
We start by describing the formal model of iterated picking-sequence mechanisms.

\section{Model}
\label{sec:model}

In this paper, we study an iterated fair allocation problem. We first introduce the static setting and useful notations. 

\vspace{0.1cm}
\noindent \textbf{Static setting.} An instance is described as a tuple $\langle \agents, \items, \succ, v \rangle$, where $\agents = \{1, \ldots, n\}$ is a set of $n$ agents, $\items = \{g_1, \ldots, g_m\}$ is a set of $m$ items, $\succ = (\succ_1, \ldots, \succ_n)$ defines, for each agent $i \in \agents$, a strict linear order (ranking) $\succ_i$ over the set of items $\items$, and $v = (v_1,\ldots,v_n)$ is a utility profile. More precisely, we assume that agent $i$'s preferences can also be represented by any additive cardinal utility function $v_i:2^\items \rightarrow \mathbb{R}_{\geq 0}$ consistent with $\succ_i$. That is, $v_i(\{g\}) > v_i(\{g'\}) \Leftrightarrow g \succ_i g'$ for all $g,g' \in \items$, and $v_i(S) = \sum_{g \in S} v_i(\{g\})$ for any bundle $S \subseteq \items$. This preference representation is equivalent to the \emph{responsive set extension} \cite{Barbera_et_al2004}. 

Let $\bids$ denote the set of strict linear orders over the items in $\items$. Throughout the paper, for brevity, we may represent a complete ranking by listing the items from most to least preferred and omitting the preference symbol. Thus, $g_{k_1}g_{k_2}\ldots g_{k_m}$ denotes
$g_{k_1} \succ_i g_{k_2} \succ_i \ldots \succ_i g_{k_m}$.
We also use capital letters to denote possibly empty sub-rankings, with juxtaposition denoting concatenation. For instance, $LgMg'R$ denotes a ranking in which the items in $L$ appear first, followed by $g$, then the items in $M$, followed by $g'$, and finally the items in $R$.

Agents are assumed to reveal their preferences by announcing a ranking over the items (potentially untruthful). Formally, we define a bid vector $\bid = (\bid_1, \ldots, \bid_n)$ where $\bid_i \in \bids$ for any agent $i \in \agents$. Given an agent $i \in \agents$, we note $\bid_{-i} \in \bids^{n-1}$ the preferences of the other agents, i.e. $\bid_{-i} = (\bid_1, \ldots, \bid_{i-1}, \bid_{i+1}, \ldots, \bid_n)$.
We are interested in mechanisms that distribute items among agents given the bid vector $\bid$. Formally, $A: \agents \rightarrow 2^{\items}$ is a complete allocation if and only if the following conditions hold: $\cup_{i \in \agents} A(i) = \items$ and $A(i) \cap A(j) = \emptyset$ for any $i, j \in \agents$, $i \neq j$. Let $\allocations$ denote the set of complete allocations.

In this paper, we focus on picking sequence mechanisms: given a sequence $\pi \in \agents^m$ of length $m$ specifying which agent picks at each step, agents pick their preferred items among those remaining. Given a sequence $\pi \in \agents^m$, we denote such a picking sequence $\mathrm{PS}_{\pi}: \bids^n \rightarrow \allocations$.
A round-robin sequence is a sequence that repeatedly cycles through the agents in a fixed order. Whenever we consider such a sequence, we assume, up to a renaming of the agents, that $\pi = (1,2,\ldots,n,1,2,\ldots)$, and we simply write $\mathrm{RR}$ instead of $\mathrm{PS}_{\pi}$ for the corresponding picking-sequence mechanism. All results in the paper are established for the round-robin setting. 

Finally, we define concepts related to agents' strategic behaviour. While they are usually defined with respect to true preferences, we define them with respect to announced preferences, which may differ from the truthful ones, as needed for the iterated setting.

\begin{definition}[strictly profitable manipulation]
    Given an instance $\langle \agents, \items, \succ, v \rangle$, a sequence $\pi  \in  \agents^m$, and a bid vector $\bid \in \bids^n$, $\indbid  \in  \bids$ is a \emph{strictly profitable manipulation} for agent $i  \in  \agents$ if $$v_i(A'(i)) > v_i(A(i))$$ where $A := \mathrm{PS}_{\pi}(\bid)$ and $A' := \mathrm{PS}_{\pi}(\indbid, \bid_{-i})$.
\end{definition}

\begin{definition}[best response]
     Given an instance $\langle \agents, \items, \succ, v \rangle$, a sequence $\pi \!  \in \!  \agents^m$, and a bid vector $\bid \! \in \! \bids^n$, $\indbid$ is a \emph{best response} to $\bid_{-i}$ for agent $i \! \in \! \agents$ if for any alternative bid $\indbid'\! \in \! \bids$, we have: $$v_i(A(i)) \geq v_i(A'(i))$$ where  $A := \mathrm{PS}_{\pi}(\indbid, \bid_{-i})$ and $A' := \mathrm{PS}_{\pi}(\indbid', \bid_{-i})$.
\end{definition}

\begin{definition}[pure Nash equilibrium]  Given an instance $\langle \agents,\items,$ $ \succ, v \rangle$ and a sequence $\pi  \in   \agents^m$, a bid vector $\bid  \in  \bids^n$ is a \emph{pure Nash equilibrium} (PNE) if, for any $i \in \agents$, $\bid_i$ is a best response to $\bid_{-i}$.
\end{definition}

\begin{definition}[strategyproof] For a given sequence $\pi \in \agents^m$, $\mathrm{PS}_{\pi}$ is said to be strategyproof if, for any instance $\langle \agents, \items, \succ, v \rangle$, the truthful bid vector $\bid =\:\succ$ is a PNE.
\end{definition}

We now recall a first result from \cite{BL-IJCAI11} for the case of identical preferences (i.e., $\succ_i =\: \succ_j$ for all $i,j \in \agents$).

\begin{theorem}[from \cite{BL-IJCAI11}]
For any instance $\langle \agents, \items, \succ, v \rangle$ with identical preferences, $\mathrm{PS}_{\pi}$ is strategyproof for any sequence $\pi \in \agents^m$.
\end{theorem}

Moreover, \cite{Aziz_et_al2017} showed that, in the general case where preferences may differ across agents, a PNE is guaranteed to exist. The equilibrium they construct is called the \emph{bluff profile}.

\begin{definition}[bluff profile]
   Consider an instance $\langle \agents, \items, \succ, v \rangle$ and a sequence $\pi \in \agents^m$. Let $\rho = (\rho_1, \ldots, \rho_m)$ denote the order in which items are selected under $\mathrm{PS}_{\pi}$  when $\bid =\: \succ$. The \emph{bluff ranking}, denoted by $\succ_{\mathrm{br}}$, is the preference order over $\items$ defined by
\[
    \rho_1 \succ_{\mathrm{br}} \rho_2 \succ_{\mathrm{br}} \dots \succ_{\mathrm{br}} \rho_m.
\]
The \emph{bluff profile} is the bid vector in which  every agent reports the bluff ranking, i.e.,
$\bid_i =\: \succ_{\mathrm{br}}$ for all $i \in \agents$.
\end{definition}

\begin{theorem}[from \cite{Aziz_et_al2017}]
    For any instance $\langle \agents, \items, \succ, v \rangle$ and any sequence $\pi \! \in \! \agents^m$, $\bid$ is a PNE if $\bid$ is the bluff profile.
\end{theorem}

The Round Robin is famous for being a simple algorithm fulfilling the following popular notion of fairness. 

\begin{definition}[Envy-freeness up to one good]
    Given an instance $\langle \agents, \items, \succ, v \rangle$, an allocation $A \in \allocations$ is \emph{envy-free up to one good (EF1)} if, for any two agents $i,j \in \agents$, either
    $v_i(A(i)) \geq v_i(A(j))$
    or there exists $g \in A(j)$ such that
    $v_i(A(i)) \geq v_i(A(j) \setminus \{g\})$.
\end{definition}

It is well known that, when provided with the agents' true preferences, Round Robin returns an EF1 allocation; that is, $A := \mathrm{RR}(\succ)$ is EF1 \cite{Caragiannis_et_al2019}. Moreover, as mentioned in the introduction, one appealing property of Round Robin is that every PNE induces an allocation that is EF1 with respect to the agents' true preferences. In other words, even if agents misreport their preferences, as long as the bid vector is a PNE, the resulting allocation is EF1 with respect to the true preferences \cite{AmanatidisBFLLR24}.

\vspace{0.1cm}
\noindent \textbf{Iterated setting.} We study an iterated version of the static setting, which allows agents to manipulate sequentially. More precisely, we start from the truthful bid profile $\bid^0=\:\succ$. At each iteration $t$, the picking-sequence mechanism constructs a complete allocation $A^t$ from $\bid^t$, and one agent $i \in \agents$ is chosen arbitrarily among those who have a strictly profitable manipulation given $\bid^t$. 
This agent is then allowed to change her bid: $\bid^{t+1}_i$ is a best response to $\bid^t_{-i}$, while $\bid^{t+1}_{-i} = \bid^t_{-i}$. The process is then repeated as long as some agent has a strictly profitable manipulation. In the sequel, we denote this iterated version of $\mathrm{PS}_{\pi}$ by $\mathrm{IPS}_{\pi}$, and by $\mathrm{IRR}$ in the round-robin setting. Moreover, we say that $\mathrm{IPS}_{\pi}$ is strategyproof if, for any instance $\langle \agents, \items, \succ, v \rangle$, $\bid^0=\:\succ$ is a PNE. In this case, no agent has a strictly profitable manipulation and $\mathrm{IPS}_{\pi}$ terminates after the initial iteration.

Note that agents are assumed to be myopic, i.e., they do not take previous iterations into account and base their manipulation solely on the current bid vector and the allocation returned at the current iteration. We also assume  that the initial bid vector is truthful, i.e., $\bid^0=\:\succ$, a common assumption in the iterated voting literature \cite{MeirPRJ17, Meir2015, Meir2018}. 
Our goal is to study whether strategic deviations starting from truthful reports eventually converge to a PNE.  

In this paper, we study the convergence of $\mathrm{IRR}$, that is, whether the dynamics eventually reach a PNE rather than entering a cycle of manipulations. We first define a cycle.

\begin{definition}[cycle]
    Given a sequence $\pi \in \agents^m$, we say that $\mathrm{IPS}_{\pi}$ admits a cycle on an instance $\langle \agents, \items, \succ, v \rangle$ if there exists a sequence of bid vectors $(\bid^0, \ldots)$ generated by $\mathrm{IPS}_{\pi}$ such that $\bid^t = \bid^{t'}$ for some iterations $t<t'$.
\end{definition}

Our objective is to identify settings in which no execution of $\mathrm{IRR}$ can enter a cycle, i.e., settings under which $\mathrm{IRR}$ necessarily converges to a PNE. In this paper, we study this problem under two informational assumptions. In Section~\ref{section: full info}, we assume that agents have full information, i.e., they know the bids of all agents (and therefore the allocation returned by $\mathrm{RR}$) at each iteration. In contrast, in Section~\ref{section: partial info}, we assume that agents have only partial information: they observe only the allocation returned by $\mathrm{RR}$ at each iteration. To describe the dynamics of $\mathrm{IRR}$ under these informational settings, we use the notation
\[
\{\bid_1, \bid_2, \ldots\}[A(1):v_1(A(1)), A(2):v_2(A(2)), \ldots]
\]
to indicate that the agents report bids $\bid_i$, receive bundles $A(i)$, and obtain utilities $v_i(A(i))$, respectively. Moreover, we write $\overset{i}{\rightarrow}$ to denote a best-response manipulation by agent $i \in \agents$.
\section{Full Information} \label{section: full info}

In this section, unless stated otherwise, we focus on $\mathrm{IRR}$ and assume that agents have full information: they know the sequence and observe the agents' bids. The first natural question is whether $\mathrm{IRR}$ converges in general under full information. 

We start with a simple observation showing that  convergence can already be guaranteed in a very restricted case: 

\begin{definition}[potential manipulator]
    Given an instance $\langle \agents, \items, \succ, v \rangle$ and a sequence $\pi \in \agents^m$, agent $i \in \agents$ is called a \emph{potential manipulator} if she appears in non-adjacent positions in the sequence.
\end{definition}

\begin{observation} \label{observation: unique manipulator}
    Given an instance $\langle \agents, \items, \succ, v \rangle$ and a sequence $\pi \in \agents^m$, $\mathrm{IPS}_\pi$ converges to a PNE if there is a unique potential manipulator. In particular, $\mathrm{IRR}$ is guaranteed to converge when $m=n+1$.
\end{observation}

Unfortunately, \cite{Aziz_et_al2017} showed that $\mathrm{IRR}$ can indeed cycle. They exhibited a cycle with $n=2$ agents and $m=10$ items. We provide a shorter cycle on a smaller instance. By Observation~\ref{observation: unique manipulator}, our example is minimal in the number of items, with $m=n+2$.

\begin{proposition} \label{proposition: cycle no restriction}
    $\mathrm{IRR}$ can cycle even for $n\!=\!2$ agents and $m\!=\!4$ items.
\end{proposition}

\begin{proof}
    Let $\langle \agents, \items, \succ, v \rangle$ be such that $\agents = \{1, 2\}$, $\items = \{a, b, c, d\}$, and $\succ = (bdac, dcba)$.
    Assume the utilities are as follows:
    \begin{center}
    \begin{tabular}{c|cccc}
        $v_1$ & $a:2$ & $b:4$ & $c:1$ & $d:3$ \\ \hline
        $v_2$ & $a:1$ & $b:2$ & $c:3$ & $d:4$
    \end{tabular}
    \end{center}
    The following execution of $\mathrm{IRR}$ enters a cycle:
    $$
    \begin{aligned}
        \{bdac, dcba\}[ba: 6, dc: 7] &\overset{1}{\rightarrow} \{dbac, dcba\}[db: 7, ac: 4] \\ 
        & \overset{2}{\rightarrow} \{dbac, bacd\}[da: 5, bc: 5] \\ 
        & \overset{1} {\rightarrow} \{bdac, bacd\}[bd: 7, ac: 4] \\
        & \overset{2}{\rightarrow} \{bdac, dcba\}[ba: 6, dc: 7]
    \end{aligned}
    $$ 

    \vspace{-0.45cm}
    \end{proof}

Proposition~\ref{proposition: cycle no restriction} shows that cycles exist (recall that this holds independently of the underlying utility function used, as long as it consistent with the responsive set extension), and hence that convergence cannot be guaranteed in general. We now introduce different restrictions on the authorized manipulations and study whether they are sufficient to ensure convergence of $\mathrm{IRR}$. Under such restrictions, both best responses and PNE are understood with respect to the corresponding subsets of admissible bids. 
In particular, we consider four such restrictions:

\begin{definition}[single-item change]
A single-item change in a ranking consists in moving one item to another position, while all others keep their relative order. Formally, ranking $\indbid' \! \in \! \bids$ can be obtained from ranking $\indbid \! \in \! \bids$ by a \emph{single-item change} if there exist an item $g \!  \in \! \items$, sub-rankings $L,M,R$, with $M \! \neq \! \emptyset$, such that 
\begin{itemize}
    \item either $\indbid=LgMR$ and $\indbid'=LMgR$ ($g$ moves down),
    \item or $\indbid=LMgR$ and $\indbid'=LgMR$ ($g$ moves up).
\end{itemize}
\end{definition}

\begin{definition}[swap] A swap in a ranking consists in exchanging the positions of two items, while the relative order of all other items remains unchanged. Formally, ranking $\indbid' \in \bids$ can be obtained from ranking $\indbid \in \bids$ by a \emph{swap} if there exist distinct items $g,g' \in \items$ and sub-rankings $L,M,R$ such that
$\indbid=LgMg'R$ and $\indbid'=Lg'MgR$.
\end{definition}

\begin{definition}[adjacent swap] An adjacent swap is a swap involving two consecutive items in a ranking. Formally,
ranking $\indbid' \in  \bids$ can be obtained from ranking $\indbid \in \bids$ by an \emph{adjacent swap} if there exist distinct items $g,g' \in \items$ and sub-rankings $L,R$ such that
$\indbid=Lgg'R$ and 
$\indbid'=Lg'gR$.
\end{definition}

\begin{definition}[item-preserving swap] 
A swap is  item-preserving if it demotes an item in the agent's bid that is initially allocated to her, while ensuring that the item remains allocated to her after the swap. Formally, given a bid vector $\bid \in \bids^n$, a sequence $\pi \in \agents^m$ and an agent $i \in \agents$, bid $\bid_i'$ can be obtained from $\bid_i$ by an \emph{item-preserving swap} if there exist distinct items $g,g' \in \items$ and sub-rankings $L,M,R$ such that:
\begin{itemize}
    \item $\bid_i=LgMg'R$ and $\bid_i'=Lg'MgR$,
    \item $g \in \mathrm{PS}_\pi(\bid)(i)$ and $g \in \mathrm{PS}_\pi(\bid_i',\bid_{-i})(i)$.
\end{itemize}
\end{definition}

We now illustrate item-preserving swaps with an example.

\begin{example}
Consider an instance $\langle \agents, \items, \succ, v \rangle$ with $\agents = \{1, 2\}$, $\items = \{a, b, c, d\}$, and $\succ = (abcd, bcda)$. Let $\bid =\: \succ$ be the truthful bid vector. The allocation $A \! := \! \mathrm{RR}(\bid)$ is given by $A(1)\! = \! \{a,c\}$ and $A(2)\! = \! \{b,d\}$. 
Assume agent $1$ manipulates by reporting a  bid $\bid_1'$:
\begin{itemize}
    \item $\bid_i' = dbca$ is a swap, but not an item-preserving swap. Indeed, letting $A' := \mathrm{RR}(\bid_1', \bid_2)$, we have $A'(1) = \{d,c\}$: item $a$ is demoted and no longer allocated to agent $1$.
    
    \item $\bid_1' = acbd$ is an adjacent swap, but not an item-preserving swap, since the demoted item $b$ is not initially allocated to agent $1$, i.e., $b \notin A(1)$;
    
    \item $\bid_1' = bacd$ is an item-preserving swap. Indeed, letting $A' := \mathrm{RR}(\bid_1', \bid_2)$, we have $A'(1) = \{a,b\}$: item $a$ is demoted but remains allocated to agent $1$.
\end{itemize}
\end{example}

Note that these restrictions do not limit the space of possible manipulations to the same extent. In the unrestricted setting, there are $\Theta(m!)$ possible manipulations. In contrast, there are $\Theta(m^2)$ possible manipulations under both single-item changes and swaps, and $O(m^2)$ under item-preserving swaps.
 Finally, adjacent swaps are the most restrictive, with $\Theta(m)$ possible manipulations.

Note also that a PNE under a restriction need not be a PNE in the unrestricted setting. Although every unrestricted PNE is EF1 with respect to the true preferences \cite{AmanatidisBFLLR24}, this need not hold under our restrictions. Proposition~\ref{proposition: item-preserving swap not EF1} shows this for item-preserving swaps, and Appendix~\ref{appendix: not EF1} for the three other types of manipulation.

\begin{proposition} \label{proposition: item-preserving swap not EF1} 
A PNE under item-preserving swaps may induce a non-EF1 allocation under $\mathrm{RR}$ with respect to the true preferences, and $\mathrm{IRR}$ may converge to such an equilibrium.
\end{proposition}

\begin{proof}
    Consider an instance $\langle \agents, \items, \succ, v \rangle$ such that $\agents = \{1, 2, 3\}$, $\items = \{a, b, c, d, e, f\}$, and $\succ = (abecfd, feadcb, bcfead)$.
    Assume the utilities are as follows:
    \begin{center}
    \begin{tabular}{c|cccccc}
        $v_1$ & $a:6$ & $b:5$ & $c:3$ & $d:1$ & $e:4$ & $f:2$ \\ \hline
        $v_2$ & $a:4$ & $b:1$ & $c:2$ & $d:3$ & $e:5$ & $f:6$  \\ \hline
        $v_3$ & $a:2$ & $b:6$ & $c:5$ & $d:1$ & $e:3$ & $f:4$
    \end{tabular}
    \end{center}
    
We restrict manipulations to item-preserving swaps and give an execution that leads to a PNE whose resulting allocation is not EF1 with respect to $\succ$:
    $$
    \begin{aligned}
        \{abecfd, &feadcb, bcfead\}[ae: 10, fd: 9, bc: 11]  \\
        &\overset{1}{\rightarrow} \{baecfd, feadcb, bcfead\}[ba: 11, fe: 11, cd: 6] \\
        &\overset{3}{\rightarrow} \{baecfd, feadcb, befcad\}[ba: 11, fd: 9, ec: 8] \\
        &\overset{2}{\rightarrow} \{baecfd, aefdcb, befcad\}[bc: 8, af: 10, ed: 4] 
    \end{aligned}
    $$
   The final equilibrium induces  allocation $A$ defined by $A(1)=\{b,c\}$, $A(2)=\{a,f\}$, and $A(3)=\{e,d\}$. Moreover, $v_3(A(3))=4 < \min_{g\in A(1)} v_3(A(1)\setminus\{g\})=5$. Hence, the resulting allocation is not EF1 with respect to the true preferences.
\end{proof}

We study convergence under these restrictions, first showing that neither single-item changes nor swaps guarantee convergence.

\begin{proposition}
    $\mathrm{IRR}$ may cycle even when manipulations are restricted to single-item changes.
\end{proposition}

\begin{proof}
    Consider an instance $\langle \agents, \items, \succ, v \rangle$ with $\agents = \{1,2\}$ and $\items = \{a,b,c,d\}$. Let $\succ = (adcb,badc)$, with utilities given by:
    \begin{center}
    \begin{tabular}{c|cccc}
        $v_1$ & $a:4$ & $b:1$ & $c:2$ & $d:3$ \\ \hline
        $v_2$ & $a:3$ & $b:4$ & $c:1$ & $d:2$
    \end{tabular}
    \end{center}
    
    \noindent The following execution of $\mathrm{IRR}$ enters a cycle:
    $$
    \begin{aligned}
    \{adcb, badc\}[ad: 7, bc: 5] 
        &\overset{2}{\rightarrow} \{adcb, dbac\}[ac: 6, db: 6] \\
        &\overset{1}{\rightarrow} \{dacb, dbac\}[da: 7, bc: 5] \\
        &\overset{2}{\rightarrow} \{dacb, dacb\}[dc: 5, ab: 7] \\
        &\overset{1}{\rightarrow} \{cdab, dacb\}[ca: 6, db: 6] \\
        &\overset{2}{\rightarrow} \{cdab, bdac\}[cd: 5, ba: 7] \\
        &\overset{1}{\rightarrow} \{dacb, bdac\}[da: 7, bc: 5] \\
        &\overset{2}{\rightarrow} \{dacb, dacb\}[dc: 5, ab: 7]
        \end{aligned}
    $$
\end{proof}

\begin{proposition}
    $\mathrm{IRR}$ may cycle even when manipulations are restricted to swaps.
\end{proposition}

\begin{proof}
   Consider an instance $\langle \agents, \items, \succ, v \rangle$ with $\agents = \{1,2\}$, $\items = \{a,b,c,d\}$, and $\succ = (dbac,bcda)$, with utilities given by:
    \begin{center}
    \begin{tabular}{c|cccc}
        $v_1$ & $a:2$ & $b:3$ & $c:1$ & $d:4$ \\ \hline
        $v_2$ & $a:1$ & $b:4$ & $c:3$ & $d:2$
    \end{tabular}
    \end{center}
    
    \noindent The following execution of $\mathrm{IRR}$ enters a cycle:
    $$
    \begin{aligned}
    \{dbac, bcda\}[da: 6, bc: 7] 
        &\overset{1}{\rightarrow} \{bdac, bcda\}[bd: 7, ca: 4] \\
        &\overset{2}{\rightarrow} \{bdac, dcba\}[ba: 5, dc: 5] \\
        &\overset{1}{\rightarrow} \{dbac, dcba\}[db: 7, ca: 4] \\
        &\overset{2}{\rightarrow} \{dbac, bcda\}[da: 6, bc: 7]
    \end{aligned}
    $$
\end{proof}

While $\mathrm{IRR}$ does not converge under single-item change and swap, we show that it converges for any number of agents and items under item-preserving swaps. We start by showing it is true for $n=2$.

\begin{theorem}
When $n=2$, $\mathrm{IRR}$ reaches a PNE under item-preserving swaps after at most $\frac{m(m-1)}{2}$ manipulations.
\end{theorem}

\begin{proof}
Consider an instance $\langle \agents, \items, \succ, v \rangle$ with
$|\agents|=2$. Let $\bid^t$ be the bid profile at iteration $t$, and let
$A^t=\mathrm{RR}(\bid^t)$. Without loss of generality, assume that agent~$1$
performs the following strictly profitable item-preserving swap:
\[
\bid_{1}^{t}=L\,g'\,M\,g\,R
\quad\longrightarrow\quad
\bid_{1}^{t+1}=L\,g\,M\,g'\,R.
\]
We have $\bid_2^{t+1}=\bid_2^t$ and
$A^{t+1}=\mathrm{RR}(\bid^{t+1})$. Since the swap is item-preserving, we have: $g' \in A^{t}(1)\cap A^{t+1}(1)$.

We first show that $g$ is ranked above $g'$ in $\bid_2^t$. Suppose, for
contradiction, that $g'$ is ranked above $g$ in $\bid_2^t$. 
Since the
manipulation is strictly profitable, $A^t(1)\neq A^{t+1}(1)$. 
Consider the executions of the round-robin sequence at iterations $t$ and
$t+1$, and let $\tau$ be the first turn at which the selected items differ.
Until turn $\tau$, the same items are selected in both executions, so the same items are available at turn $\tau$. Since agent~$2$ has the same bid in both iterations, she would make the same choice from the same set of available items. Hence, the agent acting at turn $\tau$ must be agent~$1$.

Moreover, neither $g$ nor $g'$ can have been selected before turn $\tau$.
Indeed, suppose first that $g'$ had already been selected. Since the choices
coincide before $\tau$ and $g'$ is allocated to agent~$1$ in both
executions, agent~$1$ must have selected it in both. However, in iteration $t+1$, $g$ is ranked above $g'$ in $\bid^{t+1}_1$.
Thus, if agent~$1$ selects $g'$ before turn $\tau$, then $g$ must already
have been selected. Since the two executions coincide before $\tau$, $g$
must also have been selected before $\tau$ in iteration $t$. But it could
not have been selected by agent~$1$ in iteration $t$, since $g'$ was still
available and is ranked above $g$ in $\bid^t_1$. Nor could it have been
selected by agent~$2$, since, by assumption, $g'$ is ranked above $g$ in
$\bid^t_2$. This is a contradiction. Similarly, $g$
cannot have been selected before $\tau$: agent~$1$ would prefer the still
available item $g'$ in iteration $t$, while agent~$2$ ranks $g'$ above $g$ in
$\bid^t_2$.

Thus, both $g$ and $g'$ are available at turn $\tau$, and agent~$1$ selects
$g'$ in iteration $t$ and $g$ in iteration $t+1$. Let $\tau'>\tau$ be the turn at which agent~$1$ selects $g'$ in iteration
$t+1$. Such a turn exists since the swap is item-preserving. Immediately
after turn $\tau$, the sets of remaining items in the two executions differ
only in that $g$ remains available at iteration $t$, whereas $g'$ remains
available at step $t+1$. 

We show that all choices coincide from turn $\tau+1$ to turn $\tau'-1$.
Suppose inductively that, at some such turn, the remaining sets are
$C\cup\{g\}$ at iteration $t$ and $C\cup\{g'\}$ at iteration $t+1$.
If agent~$2$ acts, she does not select $g'$ at iteration $t+1$, since $g'$
is eventually allocated to agent~$1$. Hence, she selects some $g''\in C$
ranked above $g'$ in $\bid_2^t$. By assumption, $g'$ is ranked above $g$,
so $g''$ is also ranked above $g$, and therefore agent~$2$ selects the same item $g''$
at iteration $t$. If agent~$1$ acts before turn $\tau'$, she selects some $g''\in C$ rather
than $g'$ at iteration $t+1$ (by definition of $\tau'$). Thus, $g''$ is
ranked above $g'$ in $\bid_1^{t+1}$, and therefore belongs to $L$ or $M$.
In either case, $g''$ is ranked above $g$ in $\bid_1^t$. Moreover, since
the relative order of the items in $C$ is the same in $\bid_1^t$ and
$\bid_1^{t+1}$, $g''$ is the highest-ranked item of $C$ in both bids.
Hence, agent~$1$ also selects $g'$ from $C\cup\{g\}$ at iteration $t$.
Therefore, the same item is selected in both executions at every turn
between $\tau+1$ and $\tau'-1$.

At turn $\tau'$, agent~$1$ selects $g'$ in iteration $t+1$. Hence, every
remaining item common to both executions is ranked below $g'$ in
$\bid_1^{t+1}$ and thus belongs to $R$. In iteration $t$, item $g$ is still
available and is ranked above every item in $R$, so agent~$1$ selects $g$.
The sets of remaining items are then identical, and all subsequent choices
coincide. Agent~$1$ therefore receives exactly the same bundle in both
iterations, contradicting the strict profitability of the manipulation.
Hence, $g$ is ranked above $g'$ in $\bid_2^t$.

We now use the Kendall tau distance as a potential function to prove the convergence of $\mathrm{IRR}$. For two rankings $\indbid$ and $\indbid'$, let $$
D(\indbid,\indbid')
\! = \! 
\left\{
\{x,y\}\subseteq \!  \items
\;\middle|\;
\indbid \text{ and } \indbid' \text{ rank } x \text{ and } y \text{ in opposite orders}
\right\}$$
and define $
d_{KT}(\indbid,\indbid')=|D(\indbid,\indbid')|.$ 
Thus, $d_{KT}(\indbid,\indbid')$ is the number of pairs of items ordered differently in the two rankings.

At iteration $t$, agent~$1$ ranks $g'$ above $g$ in $\bid_1^t$, while we showed that agent~$2$ ranks $g$ above $g'$ in $\bid_2^t$. After the swap, agent~$1$ ranks $g$ above $g'$ in $\bid_1^{t+1}$. Hence, the agents disagree on the relative order of $g$ and $g'$ at iteration $t$, but agree at iteration $t+1$, decreasing the Kendall tau distance by one. 

Moreover, for every $g'' \in M$, agent~$1$ initially ranks $g'$ above $g''$ and $g''$ above $g$, while after the swap she ranks $g$ above $g''$ and $g''$ above $g'$. 
If agent~$2$ ranks $g$ above $g''$ above $g'$ in $\bid_2^t$, two additional disagreements disappear. Otherwise, $g''$ is ranked either above both $g$ and $g'$, or below both. In either case, among the pairs $\{g,g''\}$ and $\{g',g''\}$, one disagreement disappears while another appears, so their contribution to the Kendall tau distance remains unchanged.  All other pairs keep the same relative order. Hence 
$d_{KT}(\bid_1^{t+1},\bid_2^{t+1})
=
d_{KT}(\bid_1^t,\bid_2^t)
-1
-2
\left|
\left\{
g''\in M
\;\middle|\;
g'' \text{ lies between } g \text{ and } g' \text{ in } \bid_2^t
\right\}
\right|.
$
Thus, every strictly profitable item-preserving swap performed by agent~$1$ decreases the Kendall tau distance by at least one. By symmetry, the same holds when agent~$2$ manipulates. Finally, the Kendall tau distance is a nonnegative integer bounded by $d_{KT}(\bid_1^0,\bid_2^0)\leq \binom{m}{2}=m(m-1)/2$. Since every manipulation decreases this distance by at least one, at most
$m(m-1)/2$ manipulations can occur. Hence, $\mathrm{IRR}$ terminates at a PNE.
\end{proof}

Thus, not only does the process converge, but the two agents' bids also become more similar throughout the process, an interesting observation in itself. However, this argument, and in particular its potential function, does not extend to $n \geq 3$. A natural attempt would be to sum the pairwise Kendall tau distances between agents' bids, but this quantity is not monotonically decreasing (as shown in Appendix~\ref{appendix: sum of KT}). We therefore turn to a different argument, which proves convergence for any $n\geq 2$, including $n=2$, but yields an $O(m!)$ bound on the number of steps rather than the quadratic bound obtained in the two-agent case.

\begin{theorem}
  When $n \geq 2$, $\mathrm{IRR}$ converges to a PNE under item-preserving swaps within $m!-1$ manipulations.
\end{theorem}

\begin{proof}
Consider an instance $\langle \agents, \items, \succ, v \rangle$ and a round-robin sequence $\pi=(\pi_1,\ldots,\pi_m)$, where $\pi_s$ denotes the agent selecting at turn $s$. For an item $g$ and agent $i$'s ranking $\bid_i$, let $\operatorname{rk}(g,\bid_i)$ denote the rank of $g$ in $\bid_i$, with rank $1$ being the highest. For every turn $s\in\{1,\ldots,m\}$, let $\phi_s(\bid)$ denote the rank, in $\bid_{\pi_s}$, of the item selected by agent $\pi_s$ at turn $s$ under bid profile $\bid$. We use the following potential function:
\[
\Phi(\bid)
=
(\phi_1(\bid),\ldots,\phi_m(\bid)).
\]
We show that every strictly profitable item-preserving swap strictly increases $\Phi$ according to the following lexicographic order: for two vectors $x,y\in\mathbb{N}^m$, we write $x\succ_{\mathrm{lex}}y$ if there exists an index $s$ such that $x_{s'}=y_{s'}$ for all $s'<s$ and $x_s>y_s$.

Consider a bid profile $\bid^t$ at some iteration $t$ of $\mathrm{IRR}$, and suppose that agent $i\in\agents$ performs the following strictly profitable item-preserving swap:
\[
\bid_i^t
=
L\,g'\,M\,g\,R
\quad\longrightarrow\quad
\bid_i^{t+1}
=
L\,g\,M\,g'\,R.
\]
Let $r_{g'}^t
=
\operatorname{rk}(g',\bid_i^t)
=
|L|+1
$ and $
r_g^t
=
\operatorname{rk}(g,\bid_i^t)
=
|L|+|M|+2.
$
After the swap, $\operatorname{rk}(g,\bid_i^{t+1})=r_{g'}^t$ and $\operatorname{rk}(g',\bid_i^{t+1})=r_g^t.$ 
Since the swap is item-preserving, $g'$ is allocated to agent $i$ both before and after the swap.

Let $\tau$ be the turn at which agent $i$ selects $g'$ in iteration $t$. The two executions coincide before turn $\tau$. Indeed, $g'$ remains available until turn $\tau$ in iteration $t$. Hence, at any earlier turn of agent $i$, she selects an item from $L$. Such an item has the same position in $\bid_i^t$ and $\bid_i^{t+1}$ and is ranked above all items whose positions are affected by the swap. Thus, assuming inductively that the same items are available, agent $i$ makes the same choice in both iterations. All other agents have unchanged bids and therefore also make the same choices. Hence, the same items are selected before turn $\tau$.

At turn $\tau$, agent $i$ selects $g'$ at rank $r_{g'}^t$ in iteration $t$. Since no item from $L$ remains, three cases are possible in iteration $t+1$:

\medskip

\noindent \textbf{(1) Agent $\boldsymbol{i}$ still selects $\boldsymbol{g'}$.}  In this case, $g$ and all items in $M$ must already have been selected, since they are ranked above $g'$ in $\bid_i^{t+1}$. As the two executions coincide before turn $\tau$, these items have also been selected in iteration $t$. Thus, after $g'$ is selected, only items from $R$ remain available, and the two executions coincide thereafter. Hence, both iterations return the same allocation, contradicting the strict profitability of the manipulation.

\medskip

\noindent \textbf{(2) Agent $\boldsymbol{i}$ selects some item $\boldsymbol{g''\in M}$.}
Since $g$ is ranked above every item of $M$ in $\bid_i^{t+1}$, item $g$ must already have been selected. Moreover, since $\operatorname{rk}(g,\bid_i^{t+1})=r_{g'}^t$,
any $g''\in M$ has rank strictly greater than $r_{g'}^t$ in $\bid_i^{t+1}$. Therefore,
\[
\phi_\tau(\bid^{t+1})
>
r_{g'}^t
=
\phi_\tau(\bid^t).
\]
Since all previous coordinates are identical, we obtain
\[
\Phi(\bid^{t+1})
\succ_{\mathrm{lex}}
\Phi(\bid^t).
\]

\medskip

\noindent \textbf{(3) Agent $\boldsymbol{i}$ selects $\boldsymbol{g}$.} Since $g'$ has rank $r_{g'}^t$ in $\bid_i^t$ and $g$ has rank $r_{g'}^t$ in $\bid_i^{t+1}$, we have \\[-15pt]
    \[
    \phi_\tau(\bid^{t+1})
    =
    \phi_\tau(\bid^t)
    =
    r_{g'}^t.
    \]
    Immediately after turn $\tau$, the sets of remaining items in the two executions differ only in that $g$ remains available in iteration $t$, whereas $g'$ remains available in iteration $t+1$.

   Let $\tau'>\tau$ be the turn at which agent $i$ selects $g'$ in iteration
$t+1$. Such a turn exists because the swap is item-preserving. We claim
that there exists a turn $q$ with $\tau<q<\tau'$ such that
\[
\phi_q(\bid^{t+1})\neq\phi_q(\bid^t).
\]

Suppose otherwise. Consider any turn $s$ with $\tau<s<\tau'$. If $\pi_s\neq i$, the acting agent has the same bid in both executions, so selecting an item of the same rank implies selecting the same item. If $\pi_s=i$, all items other than $g$ and $g'$ have the same rank in $\bid_i^t$ and $\bid_i^{t+1}$, while the common rank $r_g^t$ of $g$ in $\bid_i^t$ and $g'$ in $\bid_i^{t+1}$ cannot be selected before $\tau'$. Hence, whenever the selected ranks coincide, the same item is selected in both executions for every turn $\tau<s<\tau'$.

At turn $\tau'$, agent $i$ selects $g'$ in iteration $t+1$ and $g$ in
iteration $t$. Since both have rank $r_g^t$ in the corresponding bids,
the selected ranks are again identical. The sets of remaining items then
become identical, so the executions coincide thereafter. Thus, agent $i$
receives the same bundle in both iterations, contradicting the strict
profitability of the manipulation. Hence, such a turn $q$ exists.

By the same argument, at the beginning of turn $q$, the sets of available items have the form $C_q\cup\{g\}$ and $C_q\cup\{g'\}$ in iterations $t$ and $t+1$, respectively.

Note that turn $q$ cannot be a turn of agent $i$. Indeed, all items in $C_q$ have the same ranks in $\bid_i^t$ and $\bid_i^{t+1}$, while $g$ in $\bid_i^t$ and $g'$ in $\bid_i^{t+1}$ both have rank $r_g^t$. Hence, the available items have the same set of ranks in both executions, so agent $i$ would select an item of the same rank, contradicting the definition of $q$.

Thus, some agent $j\neq i$ acts at turn $q$. Since the swap is item-preserving, agent $j$ cannot select $g'$ in iteration $t+1$. Hence, she selects some item $g''\in C_q$. Since her bid is unchanged, $g''$ is the highest-ranked item of $C_q$ in both executions. If agent $j$ also selected $g''$ in iteration $t$, the selected ranks would be equal, contradicting the definition of $q$. Therefore, she must select $g$ in iteration $t$. It follows that $g$ is ranked above $g''$ in her bid, and thus
    \[
    \phi_q(\bid^{t+1})
    >
    \phi_q(\bid^t).
    \]
    Since all previous coordinates are identical,
    \[
    \Phi(\bid^{t+1})
    \succ_{\mathrm{lex}}
    \Phi(\bid^t).
    \]
    
We can now conclude from the three cases above. Case~(1) is impossible, while Cases~(2) and~(3) both yield $
\Phi(\bid^{t+1}) \succ_{\mathrm{lex}} \Phi(\bid^t)$.
Hence, every strictly profitable item-preserving swap strictly increases $\Phi$ lexicographically. 

It remains to bound the number of possible values of $\Phi$. Consider any turn $s\in\{1,\ldots,m\}$. At turn $s$, only $s-1$ items have already been selected. Therefore, at least one of the first $s$ items in the current bid of agent $\pi_s$ is still available, and thus
\[
1\leq \phi_s(\bid)\leq s.
\]
Hence, $\phi_s$ has at most $s$ possible values, and the total number of possible vectors $\Phi$ is at most
\[
\prod_{s=1}^m s
=
m!.
\] 
Since $\Phi$ strictly increases after every manipulation and has at most $m!$ possible values, every execution contains at most $m!-1$ manipulations. Consequently, $\mathrm{IRR}$ terminates at a PNE under item-preserving swaps after at most $m!-1$ manipulations.

\end{proof}

Finally, we turn our attention to convergence of $\mathrm{IRR}$ under adjacent swaps. We prove that it converges for $n=2, m=4$ (see Appendix~\ref{appendix: adj swap}), but may cycle for $n \geq 3$. Though we did not manage to prove it, we are fairly certain that $\mathrm{IRR}$ actually converges under adjacent swaps for $n=2$ and any $m$, and leave it as a conjecture.

\begin{restatable}{theorem}{adjswap}
    When $n=2$ and $m=4$, $\mathrm{IRR}$ converges to a PNE under adjacent swaps.
\end{restatable}

\begin{conjecture}
    When $n=2$, and $m \geq 5$, $\mathrm{IRR}$ converges to a PNE under adjacent swaps.
\end{conjecture}

\begin{proposition} 
    $\mathrm{IRR}$ may cycle under adjacent swaps, even when $n=3$ and $m=5$.
\end{proposition}

\begin{proof}
    Consider an instance $\langle \agents, \items, \succ, v \rangle$ such that $\agents = \{1, 2, 3\}$, $\items = \{a, b, c, d, e\}$, and $\succ = (deacb, badec, ebdac)$.
    Consider the following utilities:\\[-15pt]
    \begin{center}
    \begin{tabular}{c|ccccc}
        $v_1$ & $a:3$ & $b:1$ & $c:2$ & $d:5$ & $e:4$ \\ \hline
        $v_2$ & $a:4$ & $b:5$ & $c:1$ & $d:3$ & $e:2$ \\ \hline
        $v_2$ & $a:2$ & $b:4$ & $c:1$ & $d:3$ & $e:5$
    \end{tabular}
    \end{center}
    
    \noindent The following execution of $\mathrm{IRR}$ enters a cycle:
    $$
    \begin{aligned}
    \{deacb, &badec, ebdac\}[da: 8, bc: 6, e: 5]  \\
        &\overset{2}{\rightarrow} \{deacb, abdec, ebdac\}[dc: 7, ab: 9, e: 5] \\
        &\overset{1}{\rightarrow} \{edacb, abdec, ebdac\}[ed: 9, ac: 5, b: 4] \\
        &\overset{2}{\rightarrow} \{edacb, badec, ebdac\}[ea: 7, bc: 6, d: 3] \\
        &\overset{1}{\rightarrow} \{deacb, badec, ebdac\}[da: 8, bc: 6, e: 5]
    \end{aligned}
    $$
\end{proof}
\section{Partial Information} \label{section: partial info}

In Section~\ref{section: full info}, we assumed that agents had full information and could therefore observe the preferences of all other agents. In this section, we relax this assumption by considering a partial-information setting, in which agents only observe the allocation returned by the round-robin picking sequence at each iteration. We assume that, based on the observed allocation, agents infer a set of possible worlds, namely the worlds that are consistent with the observed allocation under the round-robin picking sequence. 

\begin{definition}[information set]
Given an instance $\langle \agents,\items,\succ,v\rangle$, a sequence $\pi \in \agents^m$, an agent $i\in\agents$, a bid $\indbid \in \bids$, and an allocation $A$, the \emph{information set} of agent $i$, denoted by $\mathcal{I}_i(\indbid,A)$, is the set of bid profiles of the other agents consistent with $A$:
\[
\mathcal{I}_i(\indbid,A)
=
\left\{
\bid_{-i}
\;\middle|\;
\bid \in \bids^{n},\ \bid_i = \indbid,\ \mathrm{PS}_\pi(\bid)=A
\right\}.
\]
\end{definition}

Although defined independently of the iteration, this information set
is used in the iterated setting. At iteration $t$, after observing
$A^t=\mathrm{PS}_\pi(\bid^t)$, agent $i$ considers
$\mathcal{I}_i(\bid_i^t,A^t)$ when deciding whether to manipulate.
We assume that agents are memoryless: they only consider their current
bid and the observed allocation, ignoring allocations from previous iterations. We study three behaviours guiding agents' decisions to manipulate.

\begin{definition}[behaviours under partial information]
Given an instance $\langle \agents,\items,\succ,v\rangle$, a sequence
$\pi\in\agents^m$, and an iteration $t$, let $A^t = \mathrm{PS}_\pi(\bid^t)$ be the observed
allocation and $\bid_i^t$ the current bid of agent $i\in\agents$. If agent
$i$ decides to manipulate, her manipulation $\bid_i^{t+1}$ may follow one
of the following behaviours:

\noindent $\bullet$ Agent $i$ is \emph{optimistic} if
$$   \bid_i^{t+1}
    \in
    \argmax\limits_{\indbid\in\bids}
    \;
    \max\limits_{\mathbbm{c}\in\mathcal{I}_i(\bid_i^t,A^t)}
    v_i\!\left(\mathrm{PS}_{\pi}(\indbid,\mathbbm{c})(i)\right).$$

\noindent $\bullet$  Agent $i$ is an \emph{expected-utility maximizer} if
    $$
    \bid_i^{t+1}
    \in
    \argmax_{\indbid\in\bids}
    \;
    \frac{1}{|\mathcal{I}_i(\bid_i^t,A^t)|}
    \sum_{\mathbbm{c}\in\mathcal{I}_i(\bid_i^t,A^t)}
    v_i\!\left(\mathrm{PS}_{\pi}(\indbid,\mathbbm{c})(i)\right).
    $$

\noindent $\bullet$ Agent $i$ is \emph{cautious} if $\bid_i^{t+1}
\in
\mathrm{CM}^t_i$, where 
\[
\begin{aligned}
CM_i^t=\bigl\{\indbid\in\bids\mid{}
&\forall \mathbbm{c}\in\mathcal{I}_i(\bid_i^t,A^t), v_i\!\left(\mathrm{PS}_{\pi}(\indbid,\mathbbm{c})(i)\right)
\geq v_i(A^t(i)),\\
\mbox{and } &\exists \mathbbm{c}\in\mathcal{I}_i(\bid_i^t,A^t), v_i\!\left(\mathrm{PS}_{\pi}(\indbid,\mathbbm{c})(i)\right)
> v_i(A^t(i))
\bigr\}.
\end{aligned}
\]
\end{definition}

In words, an optimistic agent evaluates each possible manipulation according to the most favorable world in her information set, and chooses one maximizing the utility she could obtain in such a world. An expected-utility maximizer instead evaluates each manipulation by its average utility over all possible worlds in her information set, assuming them to be equally likely, and chooses one maximizing this expected utility. Finally, a cautious agent only considers manipulations that are never detrimental in any possible world and strictly beneficial in at least one of them, and chooses among these a manipulation maximizing her worst-case utility. This last behaviour is more conservative, or pessimistic, and is in line with some works on iterative voting \cite{Meir2015, Surugue_Destercke2026}.

\begin{example}
Consider an instance $\langle \agents, \items, \succ, v \rangle$ with
$\agents=\{1,2\}$, $\items=\{a,b,c\}$, and $\succ=(abc,bca)$.
Let $\bid^0=\succ$ and $A^0:=\mathrm{RR}(\bid^0)$, with
$A^0(1)=\{a,c\}$ and $A^0(2)=\{b\}$.
Assume $v_1(a)=3$, $v_1(b)=2$, and $v_1(c)=1$, so that
$v_1(A^0(1))=4$. Agent $1$'s information set is $
\mathcal{I}_1(\bid_1^0,A^0)=\{\{abc\},\{bac\},\{bca\}\}$. Indeed, given her own bid $\bid_1^0=abc$ and the observed allocation $A^0$,
agent $1$ cannot distinguish between these three possible bids of agent $2$,
since each of them is consistent with the same allocation $A^0$. Consider the manipulation $\bid_1^1=bac$. Against these three possible bids
of agent $2$, agent $1$ obtains utilities $3$, $3$, and $5$, respectively. For agent $1$:
\begin{itemize}
    \item if \emph{optimistic}, the manipulation is strictly profitable, since its best-case utility is $5>v_1(A^0(1))=4$.
    \item if an \emph{expected-utility maximizer}, it is not profitable since its expected utility is $(3+3+5)/3=11/3<v_1(A^0(1)) = 4.$
    \item If \emph{cautious}, it is not admissible since it yields utility $3<v_1(A^0(1))=4$ in some possible worlds.
\end{itemize}
\end{example}

\begin{restatable}{proposition}{cycleoptimistic} \label{proposition: cycle optimistic}
    $\mathrm{IRR}$ may cycle with optimistic agents.
\end{restatable}

\begin{proof}
Consider an instance $\langle \agents, \items, \succ, v \rangle$ with $\agents = \{1,2\}$, $\items = \{a,b,c\}$, and $\succ = (bac,abc)$. Assume the utilities are as follows:\\[-15pt]
    \begin{center}
    \begin{tabular}{c|ccc}
        $v_1$ & $a:2$ & $b:3$ & $c:1$ \\ \hline
        $v_2$ & $a:3$ & $b:2$ & $c:1$
    \end{tabular}
    \end{center}
    
    The following execution of $\mathrm{IRR}$ enters a cycle:
    $$
    \begin{aligned}
    \{bac, abc\}[bc: 4, a: 3] 
        &\overset{1}{\rightarrow} \{abc, abc\}[ac: 3, b: 2] \\
        &\overset{1}{\rightarrow} \{bac, abc\}[bc: 4, a: 3]
    \end{aligned}
    $$
\end{proof}

In the proof above, the agent’s first manipulation actually decreases her utility, even though she initially expects to benefit from it. This is because she is optimistic and bases her decision on the best-case scenario, which does not correspond to the true bid. Thus, she gets lower utility and is incentivized to manipulate again, eventually leading to a cycle. The details explaining these choices are given in Appendix~\ref{appendix: missing proofs section 4}. 

Note that this example emphasizes a dynamic that could not occur in the full-information setting, namely a cycle occur with only one agent deviating.

We obtain a similar result for expected-utility maximizer. The proof is in Appendix~\ref{appendix: missing proofs section 4}.

\begin{restatable}{proposition}{cycleexpected} \label{proposition: cycle expected}
    $\mathrm{IRR}$ may cycle with expected-utility maximizer agents.
\end{restatable}

We now turn to the case of cautious agents. Before analyzing their behaviour, we introduce the notion of bluff world and adapt a known result.

\begin{definition}[bluff world]
Given an instance $\langle \agents,\items,\succ,v\rangle$, a sequence
$\pi\in\agents^m$, an agent $i\in\agents$, and a world
$\mathbbm{c}\in\bids^{n-1}$, let $\rho=(\rho_1,\ldots,\rho_m)$ denote the
order in which items are selected under
$\mathrm{PS}_\pi(\succ_i,\mathbbm{c})$. The \emph{bluff world} associated
with $\mathbbm{c}$, denoted by $\mathrm{bw}(\mathbbm{c})$, is the world
in which every agent $j\in\agents\setminus\{i\}$ reports the ranking
$\rho_1\cdots\rho_m$.
\end{definition}

Note that $\mathrm{bw}(\mathbbm{c})$ differs from the bluff profile of
\cite{Aziz_et_al2017}: only agents in $\agents\setminus\{i\}$ report the
ranking induced by $(\succ_i,\mathbbm{c})$, while agent $i$ reports
truthfully.

\begin{proposition} \label{prop: no manip}
Consider an instance $\langle \agents,\items,\succ,v\rangle$, a
round-robin sequence $\pi$, an agent $i\in\agents$, and a world
$\mathbbm{c}$. Against the bluff world $\mathrm{bw}(\mathbbm{c})$, agent
$i$ has no profitable deviation from truthful reporting.
\end{proposition}

The proof of Proposition~\ref{prop: no manip} follows exactly the proof of
Theorem~2 in \cite{Aziz_et_al2017}. Our statement differs in that agent
$i$ reports truthfully rather than playing the bluff ranking.

We now show that $\mathrm{IRR}$ is strategyproof with cautious agents.
Thus, no manipulation occurs from the truthful profile, so $\mathrm{IRR}$
terminates immediately and cannot cycle. The key step is to show that the
bluff world associated with any possible world is itself in the agent's
information set. The proof can be found in Appendix~\ref{appendix: missing proofs section 4}.

\begin{restatable}{lemma}{bluffis} \label{lemma: bluff IS}
    Consider an instance $\langle \agents,\items,\succ,v\rangle$, a
    round-robin sequence $\pi$, an agent $i\in\agents$, and a world
    $\mathbbm{c}$. Let $A=\mathrm{RR}(\succ_i,\mathbbm{c})$.
    Then
    $\mathrm{bw}(\mathbbm{c})\in\mathcal{I}_i(\succ_i,A)$.
\end{restatable}

\begin{theorem}
$\mathrm{IRR}$ is strategyproof with cautious agents.
\end{theorem}

\begin{proof}
To show that $\mathrm{IRR}$ is strategyproof, it suffices to show that
$CM_i^0=\emptyset$ for every agent $i\in\agents$. Let
$\langle \agents,\items,\succ,v\rangle$ be an instance and $\pi$ a
round-robin sequence. Let $\bid^0=\:\succ$ and
$A^0=\mathrm{RR}(\bid^0)$. Consider any agent $i\in\agents$ and any
possible manipulation $r\in\bids$.

If $r$ is not strictly profitable in any world of
$\mathcal{I}_i(\succ_i,A^0)$, then it is not a cautious manipulation.
Suppose therefore that there exists
$\mathbbm{c}\in\mathcal{I}_i(\succ_i,A^0)$ such that $v_i\!\left(\mathrm{RR}(r,\mathbbm{c})(i)\right)
>
v_i(A^0(i))$.
Let $\rho^{\mathbbm{c}}=(\rho^{\mathbbm{c}}_1,\ldots,
\rho^{\mathbbm{c}}_m)$ be the order in which items are selected under
$(\succ_i,\mathbbm{c})$. Since
$\mathbbm{c}\in\mathcal{I}_i(\succ_i,A^0)$, this execution returns
$A^0$. Moreover, by Lemma~\ref{lemma: bluff IS},
$\mathrm{bw}(\mathbbm{c})
\in
\mathcal{I}_i(\succ_i,A^0)$.

Consider the execution under $(\succ_i,\mathbbm{c})$ and the manipulated
execution under $(r,\mathbbm{c})$. Since the manipulation is strictly
profitable, the two allocations differ for agent $i$. Let $\tau$ be the
first turn at which the selected items differ. Since only agent $i$
changes her bid, $\pi_\tau=i$. Let $ g=\rho^{\mathbbm{c}}_\tau$ 
be the item selected by agent $i$ under $(\succ_i,\mathbbm{c})$, and let
$g'$ be the item she selects instead under $(r,\mathbbm{c})$.

Now consider the manipulation $r$ in $\mathrm{bw}(\mathbbm{c})$. We
first show that the same items
$\rho^{\mathbbm{c}}_1,\ldots,\rho^{\mathbbm{c}}_{\tau-1}$ are selected
before turn $\tau$. Indeed, suppose that the same items have been selected
up to some turn $s<\tau$. If agent $i$ acts at turn $s+1$, she faces the
same set of available items and reports $r$ in both executions, and
therefore makes the same choice. If some agent $j\neq i$ acts, then, since
$s+1<\tau$, the item selected under $(r,\mathbbm{c})$ is
$\rho^{\mathbbm{c}}_{s+1}$. In $\mathrm{bw}(\mathbbm{c})$, this item is
the highest-ranked available item and is therefore selected as well.
Hence, by induction, the two executions coincide up to turn $\tau-1$.

Thus, at turn $\tau$, agent $i$ faces the same set of available items and
selects $g'$ instead of $g$. Since the choices differ at turn $\tau$, we
have $\tau<m$, as only one item remains at the last turn. Moreover, as
$\pi$ is round robin and $n\geq2$, the agent acting at turn $\tau+1$ is
some $j\neq i$. In $\mathrm{bw}(\mathbbm{c})$, agent $j$ reports the
ranking $\rho^{\mathbbm{c}}_1\cdots\rho^{\mathbbm{c}}_m$.
Since $\rho^{\mathbbm{c}}_1,\ldots,\rho^{\mathbbm{c}}_{\tau-1}$ have already
been selected while $g=\rho^{\mathbbm{c}}_\tau$ remains available, $g$
is her highest-ranked available item. Hence, agent $j$ selects $g$, and
agent $i$ cannot recover it. Therefore, $\mathrm{RR}(r,\mathrm{bw}(\mathbbm{c}))(i)
\neq
A^0(i)$.

Moreover, since all agents in $\agents \setminus \{i\}$ report the bluff ranking, they can be aggregated in one, representative agent sharing the same preferences, and playing in place of those agents. Then, we obtain a two-player instance that is equivalent with our initial one. From \cite{Bouveret_Lang2014}, we know that in any two-player instance, there is a unique optimal bundle. By Proposition~\ref{prop: no manip}, truthful reporting is optimal for agent $i$ against $\mathrm{bw}(\mathbbm{c})$. Hence, by the uniqueness of the optimal bundle in the corresponding two-agent instance
\cite{Bouveret_Lang2014}, any different bundle yields strictly lower
utility. Hence, $v_i\!\left( \mathrm{RR}(r,\mathrm{bw}(\mathbbm{c}))(i)
\right) < v_i(A^0(i)). $ Thus, any manipulation that is strictly profitable in some possible world
is strictly detrimental in another world of the same information set. It
therefore cannot be cautious. Hence, $CM_i^0=\emptyset$
 for every $i\in\agents$. Therefore, no agent manipulates from the truthful profile and $\mathrm{IRR}$ is strategyproof with cautious agents.
\end{proof}
\section{Conclusion}

In this paper, we studied the convergence of the Iterated Round Robin (IRR) under different settings. In Section~\ref{section: full info}, we studied a full-information setting, and proposed different restrictions on manipulations after exhibiting a minimal cycle ($n=2, m=4$) in the unrestricted setting. We showed that for some of them (i.e. single-item change, and swap), IRR may not converge, but nevertheless identified a natural restriction (namely, item-preserving swap) under which IRR is guaranteed to converge for any number of agents, and any number of items. Moreover, we showed that for $n=2$, IRR converges with $m=4$, while leaving the question open for an arbitrary number of items; we also show that IRR may cycle for $n=3$ under the same restriction. We emphasize that the examples of cycle given throughout this section hold for any utility function consistent (in the sense of the responsive set extension) with the true ordinal preferences (we gave specific utility models only for the sake of clarity). In Section~\ref{section: partial info}, we studied a partial-information setting, where agents would infer an information set from the allocation computed by the Round Robin. We showed that two natural agents' behaviours (i.e. optimistic and expected-utility maximizer) could cycle. Here again, we stress that the counterexample given for optimistic agents holds for any utility model consistent with the true ordinal preferences, hence only the expected-utility case is specific to the choice of utilities. Finally, we showed that IRR was strategyproof when agents were cautious. In fact, this result is much more general: any picking sequence is strategyproof when agents are cautious.

This paper yields future interesting directions. For instance, we may wonder if other restrictions on manipulations may lead to convergence, and in particular, if there exists such restriction that preserve the EF1 w.r.t. the true preferences property. We may be also tempted to study additional behaviours in the partial-information setting. In particular, it could be interesting to find a less conservative behaviour than cautious that would converge to a PNE. It could also be of interest to drop the memory-less assumption of the agents to see if this changes the convergence dynamics.  

\bibliographystyle{plain}
\bibliography{./sample}

\newpage
\appendix
\section{Appendix}

\subsection{PNE is not EF1 w.r.t. the true preferences under restricted manipulations}  \label{appendix: not EF1}

In Section~\ref{section: full info}, we showed that a PNE under item-preserving-swap was not necessarily EF1 w.r.t. the true preferences. We provide counterexamples for the other restrictions: single-item change, swap, adjacent-swap. 

\begin{proposition} \label{proposition: single-item change not EF1}
    A PNE under single-item change may induce a non-EF1 allocation under $\mathrm{RR}$ with respect to the true preferences, and $\mathrm{IRR}$ may converge to such an equilibrium.
\end{proposition}

\begin{proof}
    Consider an instance $\langle \agents, \items, \succ, v \rangle$ such that $\agents = \{1, 2, 3\}$, $\items = \{a, b, c, d, e, f, g\}$, and $\succ = (abdegfc, efdgcab, abefgcd)$. 
    Assume the utilities are as follows:
    \begin{center}
    \begin{tabular}{c|ccccccc}
        $v_1$ & $a:7$ & $b:6$ & $c:1$ & $d:5$ & $e:4$ & $f:2$ & $g:3$ \\ \hline
        $v_2$ & $a:2$ & $b:1$ & $c:3$ & $d:5$ & $e:7$ & $f:6$ & $g:4$ \\ \hline
        $v_2$ & $a:7$ & $b:6$ & $c:2$ & $d:1$ & $e:5$ & $f:4$ & $g:3$
    \end{tabular}
    \end{center}
    Let $\pi$ be a round-robin sequence. We restrict manipulations to single-item change. We show that this instance may converge to a PNE that is not EF1 w.r.t. $\succ$.
    
    $$                                     
    \begin{aligned}
        \{abdegfc, &efdgcab, abefgcd\}[adc: 13, ef: 13, bg: 9]  \\                                                     
        &\overset{1}{\rightarrow} \{agbdefc, efdgcab, abefgcd\}[agd: 15, ef: 13, bc: 8] \\
        &\overset{3}{\rightarrow} \{agbdefc, efdgcab, aefgbcd\}[agc: 11, ed: 12, fb: 10] \\                                    &\overset{1}{\rightarrow} \{acgbdef, efdgcab, aefgbcd\}[acb: 14, ed: 12, fg: 7] \\                             
        &\overset{3}{\rightarrow} \{acgbdef, efdgcab, aefbcdg\}[acg: 11, ed: 12, fb: 10] \\
        &\overset{1}{\rightarrow} \{abcgdef, efdgcab, aefbcdg\}[abg: 16, ed: 12, fc: 6] \\                                     &\overset{3}{\rightarrow} \{abcgdef, efdgcab, aefbdgc\}[abc: 14, ed: 12, fg: 7] \\                             
    \end{aligned}
    $$
    
    The final equilibrium yields $v_3(A(3)) = 7 < v_3(A(1) \setminus \{a\}) = 8$. Hence, the equilibrium is not EF1 w.r.t. the true preferences.
\end{proof}                                                                                                          
\begin{proposition} \label{proposition: swap not EF1}
    A PNE under swaps may induce a non-EF1 allocation under $\mathrm{RR}$ with respect to the true preferences, and $\mathrm{IRR}$ may converge to such an equilibrium.
\end{proposition}

\begin{proof}                      
    Consider an instance $\langle \agents, \items, \succ, v \rangle$ such that $\agents = \{1, 2, 3\}$, $\items = \{a, b, c, d, e, f, g\}$, and $\succ = (bagcfed, bfeagcd, bafcedg)$. 
    Assume the utilities are as follows:
    \begin{center}
    \begin{tabular}{c|ccccccc}
        $v_1$ & $a:6$ & $b:7$ & $c:4$ & $d:1$ & $e:2$ & $f:3$ & $g:5$ \\ \hline
        $v_2$ & $a:4$ & $b:7$ & $c:2$ & $d:1$ & $e:5$ & $f:6$ & $g:3$ \\ \hline
        $v_2$ & $a:6$ & $b:7$ & $c:4$ & $d:2$ & $e:3$ & $f:5$ & $g:1$
    \end{tabular}
    \end{center}
    
    Let $\pi$ be a round-robin sequence. We restrict manipulations to  swap. We show that this instance may converge to a PNE that is not EF1 w.r.t. $\succ$.               
    $$
    \begin{aligned}
        \{bagcfed, &bfeagcd, bafcedg\}[bgd: 13, fe: 11, ac: 10]  \\
        &\overset{1}{\rightarrow} \{bcgafed, bfeagcd, bafcedg\}[bcg: 16, fe: 11, ad: 8] \\
        &\overset{3}{\rightarrow} \{bcgafed, bfeagcd, bcfaedg\}[bgd: 13, fe: 11, ca: 10] \\
        &\overset{1}{\rightarrow} \{bcagfed, bfeagcd, bcfaedg\}[bag: 18, fe: 11, cd: 6] \\
        &\overset{3}{\rightarrow} \{bcagfed, bfeagcd, dcfaebg\}[bcg: 16, fe: 11, da: 8] \\
        &\overset{1}{\rightarrow} \{bdagfec, bfeagcd, dcfaebg\}[bag: 18, fe: 11, dc: 6] \\
        &\overset{3}{\rightarrow} \{bdagfec, bfeagcd, cdfaebg\}[bdg: 13, fe: 11, ca: 10] \\
        &\overset{1}{\rightarrow} \{bgadfec, bfeagcd, cdfaebg\}[bga: 18, fe: 11, cd: 6] \\
        &\overset{3}{\rightarrow} \{bgadfec, bfeagcd, cgfaebd\}[bgd: 13, fe: 11, ca: 10] \\
        &\overset{1}{\rightarrow} \{bfadgec, bfeagcd, cgfaebd\}[bad: 14, fe: 11, cg: 5] \\
        &\overset{3}{\rightarrow} \{bfadgec, bfeagcd, egfacbd\}[bad: 14, fg: 9, ec: 7] \\
    \end{aligned}
    $$

    The final equilibrium yields $v_3(A(3)) = 7 < v_3(A(1) \setminus \{b\}) = 8$. Hence, the equilibrium is not EF1 w.r.t. the true preferences.
\end{proof}

\begin{proposition} \label{proposition: adjacent swap not EF1}
    A PNE under adjacent swaps may induce a non-EF1 allocation under $\mathrm{RR}$ with respect to the true preferences, and $\mathrm{IRR}$ may converge to such an equilibrium.
\end{proposition}

\begin{proof}
    Consider an instance $\langle \agents, \items, \succ, v \rangle$ such that $\agents = \{1, 2, 3\}$, $\items = \{a, b, c, d, e, f, g, h, i, j\}$, and the agents' preferences are $\succ = (cadgeibjfh, iajcdfhegb, jfdgbhcaie)$.
    Assume the utilities are as follows:
    \begin{center}
    \resizebox{1.0\linewidth}{!}{
    \begin{tabular}{c|cccccccccc}
        $v_1$ & $a:9$ & $b:4$ & $c:10$ & $d:8$ & $e:6$ & $f:2$ & $g:7$ & $h:1$ & $i:5$ & $j:3$ \\ \hline
        $v_2$ & $a:9$ & $b:1$ & $c:7$ & $d:6$ & $e:3$ & $f:5$ & $g:2$ & $h:4$ & $i:10$ & $j:8$ \\ \hline
        $v_2$ & $a:3$ & $b:6$ & $c:4$ & $d:8$ & $e:1$ & $f:9$ & $g:7$ & $h:5$ & $i:2$ & $j:10$
    \end{tabular}
    }
    \end{center}
    
    Let $\pi$ be a round-robin sequence. We restrict manipulations to adjacent-swap. We show that this instance may converge to a PNE that is not EF1 w.r.t. $\succ$.
    {\small
    $$
    \begin{aligned}
        &\{cadgeibjfh, iajcdfhegb, jfdgbhcaie\}[cage: 32, idh: 20, jfb: 25]  \\
        &\overset{2}{\rightarrow} \{cadgeibjfh, aijcdfhegb, jfdgbhcaie\}[cdge: 31, aih: 23, jfb: 25] \\
        &\overset{1}{\rightarrow} \{acdgeibjfh, aijcdfhegb, jfdgbhcaie\}[acge: 32, idh: 20, jfb: 25] \\
        &\overset{2}{\rightarrow} \{acdgeibjfh, ajicdfhegb, jfdgbhcaie\}[acge: 32, jih: 22, fdb: 23] \\
        &\overset{1}{\rightarrow} \{adcgeibjfh, ajicdfhegb, jfdgbhcaie\}[adce: 33, jih: 22, fgb: 22] \\
        &\overset{2}{\rightarrow} \{adcgeibjfh, ajcidfhegb, jfdgbhcaie\}[adeh: 24, jci: 25, fgb: 22] \\
        &\overset{1}{\rightarrow} \{acdgeibjfh, ajcidfhegb, jfdgbhcaie\}[acge: 32, jih: 22, fdb: 23] \\
        &\overset{2}{\rightarrow} \{acdgeibjfh, ajcdifhegb, jfdgbhcaie\}[aceh: 26, jdi: 24, fgb: 22] \\
        &\overset{1}{\rightarrow} \{cadgeibjfh, ajcdifhegb, jfdgbhcaie\}[cdge: 31, aih: 23, jfb: 25] \\
        &\overset{2}{\rightarrow} \{cadgeibjfh, ajcdfihegb, jfdgbhcaie\}[cdeh: 25, afi: 24, jgb: 23] \\
        &\overset{1}{\rightarrow} \{cadgiebjfh, ajcdfihegb, jfdgbhcaie\}[cdie: 29, afh: 18, jgb: 23] \\
        &\overset{1}{\rightarrow} \{acdgiebjfh, ajcdfihegb, jfdgbhcaie\}[acie: 30, jdh: 18, fgb: 22] \\
    \end{aligned}
    $$
    }
    The final equilibrium yields $v_2(A(2)) = 18 < v_2(A(1) \setminus \{i\}) = 19$. Hence, the equilibrium is not EF1 w.r.t. the true preferences.
\end{proof}

\subsection{Sum of Kendall tau distances is not a suitable potential function} \label{appendix: sum of KT}

In Section~\ref{section: full info}, we proved that IRR converges under safe-swap, for $n=2$ agents, showing the the Kendall tau distance of the players' bids was monotonically decreasing. To prove convergence for arbitrary $n$, a natural path is to try to extend this potential function, e.g., by considering the sum of the Kendall tau distances. Unfortunately, we show that such a potential function is not monotonically decreasing.

\begin{example}
    Consider an instance $\langle \agents, \items, \succ, v \rangle$ such that $\agents = \{1, 2, 3\}$, $\items = \{a, b, c, d, e, f\}$, and $\succ = (abcdef, adbefc, acbdef)$. Let $\bid =\: \succ$ be a truthful bid vector, and $A := \mathrm{RR}(\bid)$ such that $A(1) = \{a, b\}$, $A(2) = \{d, e\}$, and $A(3) = \{c, f\}$. 

    First, let's compute the Kendall tau distance between all pairs of vectors in $\bid$, and denote $KT(\bid_i, \bid_j)$ the Kendall tau distance between rankings $\bid_i$ and $\bid_j$. We have: $KT(\bid_1, \bid_2) = 4$, $KT(\bid_1, \bid_3) = 1$, and $KT(\bid_2, \bid_3) = 5$. 

    Now, assume that agent $3$ manipulates and announces $\bid_3' = aebdcf$, and let $\bid' = (\bid_1, \bid_2, \bid_3')$. Note that this manipulation is an item-preserving swap since it is a swap, and agent $3$ still obtains $c$. Indeed, let $A' := \mathrm{RR}(\bid')$. Then, we have $A'(3) = \{e, c\}$. Now, let's recompute the Kendall tau distance between all pairs of vectors in $\bid'$: $KT(\bid_1, \bid_2) = 4$, $KT(\bid_1, \bid_3) = 4$, and $KT(\bid_2, \bid_3) = 4$.

    Notice that under $\bid$ the sum of the Kendall tau was equal to $10$, while under $\bid'$, it is equal to $12$. Hence, the sum of Kendall tau distances can increase even after an item-preserving swap when $n=3$.
\end{example}

\subsection{Proof convergence under adjacent swap} \label{appendix: adj swap}

\adjswap*

\begin{proof}
Consider an instance $\langle \agents,\items,\succ,v\rangle$ with
$n=2$ and $m=4$. Up to renaming agents and items, let
$\pi=(1,2,1,2)$ and $\succ_1: abcd$. The initial bid profile is truthful, i.e., $\bid^0=\: \succ$.

We establish two claims. First, item $d$ always remains last in agent
$1$'s bid. Second, every profitable manipulation by agent $2$ moves
$d$ one position downwards in her bid.

\medskip

\noindent\textbf{Claim 1.}
Agent $1$ never moves $d$ upwards in her bid.

\begin{proof}
Initially, $d$ is ranked last by agent $1$. Suppose that, at some
iteration, $d$ is still last in $\bid_1$, and let $g$ be the item ranked
immediately above it. Consider the adjacent swap that exchanges $g$ and
$d$. This swap cannot affect agent $1$'s first choice, since $g$ and $d$ occupy the last two positions of her bid. Hence, it can only affect her second choice. If $g$ and $d$ are not both available at that turn, the swap has no effect on the allocation. Otherwise, agent $1$ selects $g$ before the
swap and $d$ after the swap. Since $g\succ_1 d$, the swap is not profitable.
Thus, agent $1$ never moves $d$ upwards. Since $d$ is initially ranked
last, it remains last in her bid throughout the execution.
\end{proof}

\medskip
\noindent\textbf{Claim 2.}
Every profitable manipulation by agent $2$ moves $d$ one position
downwards in her bid.

\begin{proof}
We argue by induction on the number of manipulations performed by agent
$2$. Initially, her bid is truthful. Assume that all her previous
manipulations have moved $d$ downwards. Hence, her current bid is obtained
from $\succ_2$ by moving $d$ downwards, while preserving the relative
order of all other items.

By Claim~1, $d$ is always ranked last by agent $1$. Therefore, agent $2$
is guaranteed to receive $d$: if she does not select it at turn $2$, then
at turn $3$ agent $1$ selects the other remaining item, leaving $d$ to
agent $2$ at turn $4$. Consider now an adjacent swap by agent $2$. If this swap does not change her choice at turn $2$, then the same items remain available at turn $3$, so agent $1$ makes the same choice, and consequently the same item remains for agent $2$ at turn $4$. Hence, the swap can affect agent $2$'s allocation only if it changes her choice at turn $2$.

First, suppose that the swap exchanges two items $g$ and $g'$ different
from $d$, with $g$ ranked immediately above $g'$ in the current bid.
Since previous manipulations by agent $2$ only moved $d$ downwards, the
relative order of $g$ and $g'$ is the same as in her true preferences;
hence $g\succ_2 g'$. If the swap changes agent $2$'s choice at turn $2$, then she selects $g$ before the swap and $g'$ after the swap. Since she receives $d$ in both cases, her bundle changes from $\{g,d\}$ to $\{g',d\}$. As $g\succ_2 g'$, this manipulation is strictly worse and therefore cannot
be profitable.

Second, suppose that the swap moves $d$ upwards, exchanging it with the
item $g$ immediately above it. If the swap has an effect, agent $2$
selects $g$ at turn $2$ before the swap and $d$ after the swap. Before
the swap, her final bundle is $\{g,d\}$. After the swap, her final bundle
is $\{d,g'\}$ for some item $g'$. Since $g$ was the highest-ranked
available item different from $d$, and the relative order of all such
items agrees with $\succ_2$, either $g'=g$ or $g\succ_2 g'$.
Hence, moving $d$ upwards cannot be profitable.

Therefore, any profitable adjacent swap by agent $2$ must move $d$ one
position downwards.
\end{proof}

By Claim~2, each manipulation by agent $2$ moves $d$ one position
downwards. Hence, she can manipulate at most three times. Suppose, for contradiction, that an execution is infinite. Since agent $2$ manipulates only finitely many times, there is some
iteration after which only agent $1$ manipulates. From that point on,
agent $2$'s bid is fixed, while every manipulation by agent $1$ strictly
increases her utility. Since there are only finitely many possible
allocations, agent $1$ cannot manipulate infinitely often, a contradiction. 

Therefore, every execution terminates. At termination, no agent has a
profitable adjacent swap, and hence the resulting bid profile is a pure
Nash equilibrium under adjacent swaps.
\end{proof}

\subsection{Missing proofs from Section~\ref{section: partial info}} \label{appendix: missing proofs section 4}

\cycleoptimistic*

\begin{proof}
Consider an instance $\langle \agents, \items, \succ, v \rangle$ with $\agents = \{1,2\}$, $\items = \{a,b,c\}$, and $\succ = (bac,abc)$. Assume the utilities are as follows:
    \begin{center}
    \begin{tabular}{c|ccc}
        $v_1$ & $a:2$ & $b:3$ & $c:1$ \\ \hline
        $v_2$ & $a:3$ & $b:2$ & $c:1$
    \end{tabular}
    \end{center}
    
    The following execution of $\mathrm{IRR}$ enters a cycle:    
    $$
    \begin{aligned}
    \{bac, abc\}[bc: 4, a: 3] 
        &\overset{1}{\rightarrow} \{abc, abc\}[ac: 3, b: 2] \\
        &\overset{1}{\rightarrow} \{bac, abc\}[bc: 4, a: 3]
    \end{aligned}
    $$

    We detail the reasoning of agent~1 at each step. 

    \paragraph{First step.} Agent~1 announces $bac$ and receives $\{b, c\}$, while agent~2 receives $a$. Agent~1 thus infers that agent~2 prefers $a$ to $c$, i.e., the set of possible worlds is $\mathrm{IS}_1(A^0) = \{abc, acb, bac\}$.
    The set of possible manipulations is $\{abc, acb, bca, cab, cba\}$. 
    
    If $1$ announces $bca$, it yields $4$ of utility in each possible world. If $1$ announces $cab$ or $cba$, it yields $4$ of utility in worlds $abc$ and $acb$, and $3$ in world $bac$. Finally, if it announces $abc$ or $acb$, it yields $3$ in worlds $abc$ and $bac$, and $5$ in world $acb$. 
    
    Agent~$1$ being optimistic will choose to announce $abc$ (choosing between $abc$, and $acb$ arbitrarily).
    
    \paragraph{Second step.} Agent~1 announces $abc$ and receives $\{a, c\}$, while agent~2 receives $b$. Agent~1 now infers that agent~2 prefers $b$ to $c$, i.e., $\mathrm{IS}_1(A^1) = \{abc, bac, bca\}$. The set of possible manipulations is $\{acb, bac, bca, cab, cba\}$.
    
    If $1$ announces $acb$, it yields $3$ of utility in each possible world. If $1$ announces $cab$ or $cba$, it yields $3$ of utility in worlds $bac$ and $bca$, and $4$ in world $abc$. Finally, if it announces $bac$ or $bca$, it yields $4$ in worlds $abc$ and $bac$, and $5$ in world $bca$. 
    
    Agent~$1$ being optimistic will choose to announce $bac$, and we are back to the initial profile.
\end{proof}

We give an example to prove Proposition~\ref{proposition: cycle expected}

\cycleexpected*

\begin{proof}
Consider an instance $\langle \agents, \items, \succ, v \rangle$ with $\agents = \{1,2\}$, $\items = \{a,b,c,d,e\}$, and $\succ = (edcab,edbac)$. Assume the utilities are:
    \begin{center}
    \begin{tabular}{c|ccccc}
        $v_1$ & $a:2$ & $b:1$ & $c:3$ & $d:4$ & $e:5$ \\ \hline
        $v_2$ & $a:2$ & $b:3$ & $c:1$ & $d:4$ & $e:5$
    \end{tabular}
    \end{center}
    
    The following execution of $\mathrm{IRR}$ enters a cycle:
    $$
    \begin{aligned}
    \{edcab, edbac\}[eca: 10, db: 7] 
        &\overset{1}{\rightarrow} \{decab, edbac\}[dca: 9, eb: 8] \\
        &\overset{1}{\rightarrow} \{edcab, edbac\}[eca: 10, db: 7]
    \end{aligned}
    $$
\end{proof}

Finally, we give a proof of Lemma~\ref{lemma: bluff IS}

\bluffis*

\begin{proof}
Let $\rho=(\rho_1,\ldots,\rho_m)$ denote the order in which items are
selected under $(\succ_i,\mathbbm{c})$. By definition, in
$\mathrm{bw}(\mathbbm{c})$, every agent $j\neq i$ reports the ranking
$\rho_1\cdots\rho_m$, while agent $i$ reports $\succ_i$. We prove by
induction that the same item is selected at every turn under
$(\succ_i,\mathbbm{c})$ and $(\succ_i,\mathrm{bw}(\mathbbm{c}))$.

At turn $1$, if agent $i$ is the picker, she reports $\succ_i$ in both
executions and therefore selects the same item, namely $\rho_1$. If some
agent $j\neq i$ is the picker, she selects $\rho_1$ under
$(\succ_i,\mathbbm{c})$ by definition of $\rho_1$. In
$\mathrm{bw}(\mathbbm{c})$, $\rho_1$ is her highest-ranked item, so she
selects it as well.

Suppose now that the same items $\rho_1,\ldots,\rho_\tau$ have been
selected in both executions, for some $\tau\geq1$. Hence, the same items
are available at turn $\tau+1$. If agent $i$ is the picker, she reports
$\succ_i$ in both executions and therefore selects the same item. If some
agent $j\neq i$ is the picker, she selects $\rho_{\tau+1}$ under
$(\succ_i,\mathbbm{c})$ by definition of $\rho_{\tau+1}$. In
$\mathrm{bw}(\mathbbm{c})$, since $\rho_1,\ldots,\rho_\tau$ have already
been selected, $\rho_{\tau+1}$ is the highest-ranked available item.
Hence, she also selects $\rho_{\tau+1}$.

Thus, by induction, the same item is selected at every turn in both
executions. Therefore, $\mathrm{RR}(\succ_i,\mathrm{bw}(\mathbbm{c}))
= \mathrm{RR}(\succ_i,\mathbbm{c}) = A$,
and hence $
\mathrm{bw}(\mathbbm{c})\in\mathcal{I}_i(\succ_i,A)$.
\end{proof}

\end{document}